\documentclass[conference]{IEEEtran}
\IEEEoverridecommandlockouts

\usepackage{graphicx}
\usepackage{subfigure}
\usepackage{amsmath,amsthm,amsfonts,amssymb}
\usepackage{cite}
\usepackage{bm}
\usepackage{bbm}
\usepackage{url}
\usepackage{array}
\usepackage{color,soul}
\usepackage{multirow}
\usepackage{booktabs}
\usepackage[table,xcdraw]{xcolor}
\usepackage{enumitem}
\usepackage{verbatim}
\usepackage{amsmath}
\usepackage{mathrsfs} 
\usepackage{xcolor}

\theoremstyle{plain}
\newtheorem{thm}{Theorem}
\newtheorem{lem}[thm]{Lemma}

\newtheorem{rem}{Remark}
\newtheorem{sty1}{Theorem}
\newtheorem{defi}[sty1]{Definition}

\usepackage[english]{babel}
\usepackage{algorithm}
\usepackage[noend]{algpseudocode}

\allowdisplaybreaks[4]

\begin{document}
\title{Fractional Fourier Domain PAPR Reduction}

\author{Yewen Cao,~Yulin Shao,~Rose Qingyang Hu

\thanks{Y. Cao and Y. Shao are with the State Key Laboratory of Internet of Things for Smart City and the Department of Electrical and Computer Engineering, University of Macau, Macau S.A.R. (E-mails: \{yc47409,ylshao\}@um.edu.mo).}
\thanks{R. Q. Hu is with the Department of Electrical and Computer Engineering, Utah State University, Logan, UT 84322 USA (e-mail:
rose.hu@usu.edu).}
}

\maketitle
\begin{abstract}
High peak-to-average power ratio (PAPR) has long posed a challenge for multi-carrier systems, impacting amplifier efficiency and overall system performance. This paper introduces dynamic angle fractional Fourier division multiplexing (DA-FrFDM), an innovative multi-carrier system that effectively reduces PAPR for both QAM and Gaussian signals with minimal signaling overhead. DA-FrFDM leverages the fractional Fourier domain to balance PAPR characteristics between the time and frequency domains, achieving significant PAPR reduction while preserving signal quality. Furthermore, DA-FrFDM refines signal processing and enables one-tap equalization in the fractional Fourier domain through the simple multiplication of time-domain signals by a quadratic phase sequence. Our results show that DA-FrFDM not only outperforms existing PAPR reduction techniques but also retains efficient inter-carrier interference (ICI) mitigation capabilities in doubly dispersive channels.
\end{abstract}


\section{Introduction}
Multi-carrier communication systems, such as orthogonal frequency division multiplexing (OFDM), are renowned for their high spectral efficiency and robustness against multi-path fading, making them a cornerstone in modern broadband wireless systems \cite{stuber2004broadband,rahmatallah2013peak,shao2021federated}. By dividing the available spectrum into several orthogonal subcarriers, these systems effectively handle high data rate transmissions over hostile wireless channels. This technique also simplifies the equalization process at the receiver end, which is crucial for maintaining signal integrity in multipath channels.

Despite the numerous advantages, multi-carrier systems face a notable drawback: high peak-to-average power ratio (PAPR) \cite{rahmatallah2013peak}. High PAPR in the transmitted signal not only complicates the amplifier design but also forces the use of expensive and power-inefficient linear power amplifiers to avoid non-linear distortion. This issue is particularly critical as it directly impacts the energy efficiency and operational cost of wireless communication systems, making PAPR reduction a pivotal area of research \cite{rahmatallah2013peak,overview_PAPR}. The emergence of new communication paradigms, such as joint source-channel coding (JSCC) \cite{JSCC}, semantic communication \cite{bourtsoulatze2019deep,TSC}, and optical communications \cite{2409.11928,op2}, has introduced discrete-time Gaussian-amplitude signals, which are even more sensitive to PAPR-related issues. Consequently, these advancements amplify the need for more efficient and robust PAPR reduction techniques.

In addressing the high PAPR in multi-carrier systems, a variety of reduction techniques have been developed \cite{rahmatallah2013peak,shao2022semantic}. Traditional methods include clipping and filtering \cite{clipping}, which, while straightforward, can cause notable signal distortion. Selective Mapping (SLM) \cite{SLM} and Partial Transmitted Sequence (PTS) \cite{PTS} offer more controlled approaches by manipulating the phase of the signal to minimize PAPR. Coding techniques \cite{coding}, on the other hand, provide a trade-off between bandwidth efficiency and system complexity. Other strategies, such as tone reservation\cite{TR} and tone injection\cite{TI}, focus on utilizing reserved subcarriers to adjust the peak amplitude, potentially reducing the spectral efficiency or necessitating substantial additional signaling. 

{\it Contributions:} In this paper, we introduce dynamic angle fractional Fourier division multiplexing (DA-FrFDM), a novel multi-carrier system that significantly reduces PAPR for both QAM and Gaussian signals with minimal signaling overhead. The core insight driving our approach is that PAPR is mainly determined by the maximum amplitude of the signal within a given domain at a certain average power level. The PAPR characteristics of a signal generally exhibit a dual relationship between the time and frequency domains. Efforts to minimize PAPR in one domain often result in an increase in the other. For example, while QAM signals demonstrate advantageous PAPR properties in the time domain, they perform suboptimally in the frequency domain. This inherent duality has steered our exploration towards an intermediate domain -- the fractional Fourier domain -- which lies between the time and frequency domains. This novel approach opens up a more effective avenue for optimizing PAPR performance, harnessing the unique properties of the discrete fractional Fourier transform (DFrFT) \cite{candan2000discrete,dfrft}. Our DA-FrFDM system presents three key advantages:
\begin{itemize}[leftmargin=0.45cm]
    \item Low PAPR. DA-FrFDM achieves substantial PAPR reduction, outperforming existing techniques like clipping, PTS, and SLM, and efficiently handles both QAM and Gaussian signals.
    \item Simple equalization. Unlike existing DFrFT-based systems that require complex channel manipulation, DA-FrFDM supports one-tap equalization in the fractional Fourier domain by simply multiplying the time-domain samples with a quadratic phase sequence.
    \item Effective inter-carrier interference (ICI) mitigation in doubly dispersive channels. DA-FrFDM offers an inherent advantage over traditional OFDM systems by effectively mitigating ICI in frequency-selective and fast fading channels. We demonstrate that DA-FrFDM is capable of reducing PAPR without compromising its ICI mitigation capability.
\end{itemize}


\section{DA-FrFDM System Overview}\label{sec:II}
This section presents the structural framework of DA-FrFDM.  An illustrative diagram of the framework is provided in Fig.~\ref{fig:system} to support the functional aspects described below.

\begin{figure}[t]
  \centering
  \includegraphics[width=0.85\columnwidth]{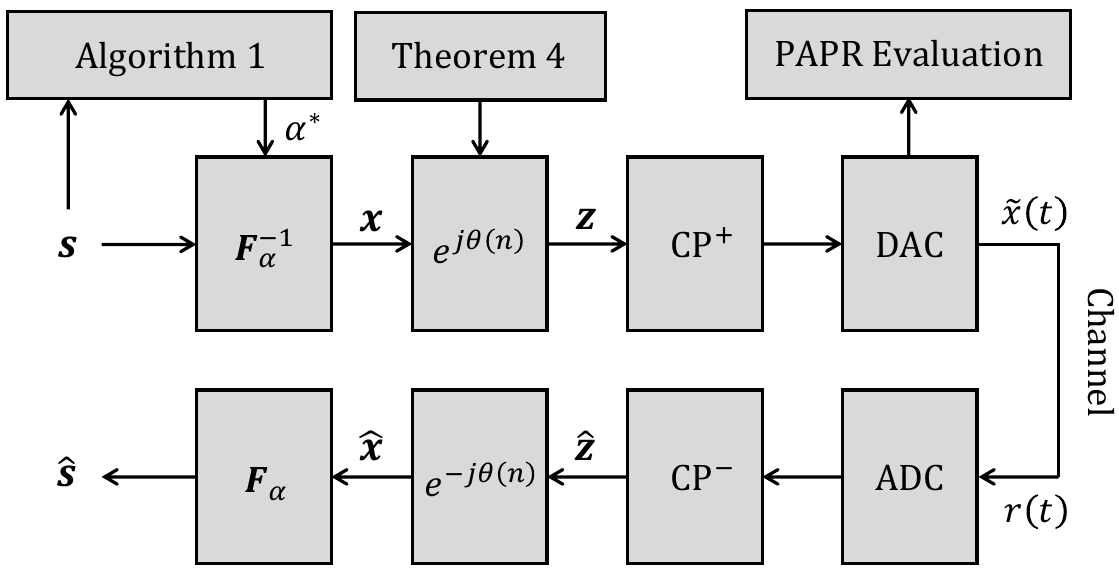}\\
  \caption{The structural framework of DA-FrFDM.}
\label{fig:system}
\end{figure}

Let $\bm{s}=\big[s[0],\allowbreak s[1],\allowbreak \cdots,\allowbreak s[N-1]\big]^\top$ represent a block of source symbols to be transmitted, where $s[k]$ can either be discrete-amplitude symbols, such as QAM symbols, or continuous-amplitude symbols, such as complex Gaussian symbols. 
In our DA-FrFDM system, these symbols are initially placed in the fractional Fourier domain, and subsequently transformed to the time domain using an inverse DFrFT (IDFrFT), $\bm{F}_{\alpha}^{-1}$. The IDFrFT is characterized by an order $a\in\left(-2,2\right]$, or equivalently, a shift angle $\alpha\triangleq\frac{a\pi}{2}\in\left(-\pi,\pi\right]$.


When $\alpha= \ell\pi$, $\ell \in \mathbb{Z}$, IDFrFT degenerates to single carrier frequency domain equalization (SC-FDE), which does not support the multiplexing capabilities required for multi-user systems. Additionally, due to the inherent symmetry and periodicity of the DFrFT, we set $\alpha\in(0,\pi) $. The time-domain samples of the DA-FrFDM block, denoted by $\bm{x}\triangleq\big[x[0],\allowbreak x[1],\allowbreak \cdots,\allowbreak x[N-1]\big]^\top=\bm{F}^{-1}_{\alpha}\bm{s}$, can be written as\cite{dfrft}
\begin{equation}\label{eq:II1}
\begin{aligned}
\hspace{1.2cm}x[n] &= \sqrt{\frac{\sin\alpha + j\cos\alpha}{N}} \times \\
&\hspace{-2cm} \sum\limits_{k=0}^{N-1} s[k] e^{-\frac{j}{2}n^2\cot\alpha T_s^2} e^{-\frac{j}{2}k^2\cot\alpha(\Delta u)^2} e^{j2\pi nk/N}, \\
&\hspace{-1cm} 0\leq k\leq N-1,0\leq n\leq N-1,
\end{aligned}
\end{equation}
where $T_s$ is the sampling interval of $x[n]$ and $\Delta u$ is the sampling interval of $X[k]$, with the relationship $\Delta u T_s=\frac{2\pi\sin{\alpha}}{N}$. When $\alpha=\frac{\pi}{2}$, the fractional Fourier domain reverts to the frequency domain, and the DA-FrFDM system degenerates to the traditional OFDM system.

Inspired by the dual relationship between PAPR in the time and frequency domains, DA-FrFDM identifies an optimal fractional domain by determining the most effective angle, denoted as $\alpha^*$, to minimize the PAPR for each realization of the data symbol block $\bm{s}$. The methodology for determining $\alpha^*$ and its impact on system performance are detailed later in Section \ref{sec:III}.

After transforming the symbol block $\bm{s}$ into the time-domain samples $\bm{x}$, DA-FrFDM applies a phase term $e^{j\theta(n)}$ to each sample and adds a cyclic prefix (CP). The purpose of the phase term $e^{j\theta(n)}$ is to facilitate the use of one-tap equalization in the fractional Fourier domain at the receiver. This refined design avoids the complex channel manipulation in existing DFrFT-based systems \cite{amir}, and streamlines the signal processing. The phase function $\theta(n)$ is given later in Theorem \ref{thm:cir_thm}. 

Upon passing through the digital-to-analog converter (D\-AC), the DA-FrFDM block is converted into a continuous-time format, represented as $\widetilde{x}(t)$. The signal is then transmitted over a wireless channel to the receiver, and the received signal $r(t)$ can be written as
\begin{equation}\label{eq:yt}
    r(t) =  h(t) \star \widetilde{x}(t)+ n(t),
\end{equation}
where $\star$ denotes the linear convolution operation and $h(t)$ characterizes the channel response. $n(t)$ is additive white Gaussian noise (AWGN).


At the receiver, the signal $r(t)$ undergoes analog-to-digital conversion (ADC), removal of the CP, and phase term $e^{-j\theta(n)}$. The resulting time-domain samples, denoted by $\bm{\widehat{x}}=\big[\widehat{x}[0],\allowbreak \widehat{x}[1],\allowbreak \cdots,\allowbreak \widehat{x}[N-1]\big]^\top$, are converted back to the original fractional Fourier domain using a DFrFT. This transformation produces the reconstructed data symbols $\bm{\widehat{s}}=\bm{F}_{\alpha}\bm{\widehat{x}}$. Finally, $\bm{\widehat{s}}$ is equalized before the final decoding process to rectify the channel impairments.

\begin{rem}
DFrFT can be implemented in several ways, with the two primary forms based on sampling and eigenvalue decomposition. In this paper, our analysis predominantly focuses on the sampling-based form, which allows for a closed-form representation of signals, facilitating our analytical work. Nonetheless, our DA-DFrFT system is also compatible with the eigenvalue decomposition-based form.
\end{rem}


\section{PAPR of DA-FrFDM}\label{sec:III}
The core functionality of the DA-DFrFDM system is its ability to dynamically adjust the angle within the fractional Fourier domain, aiming to identify the optimal angle that minimizes the PAPR. To this end, this section analyzes the PAPR and identify the optimal angle for DA-DFrFDM.

\subsection{PAPR of DA-FrFDM}
PAPR is defined as the ratio of the maximum instantaneous power of a signal to its average power. 
In our DA-FrFDM system, the phase term $e^{j\theta(n)}$ at the transmitter does not affect the PAPR. CP is a replication of the end part of the multi-carrier block, and does not alter the peak instantaneous power and only marginally affects the signal's average power. Therefore, the CP has a minimal impact on the system's PAPR. To effectively evaluate the PAPR, we can focus on the time-domain samples right after the inverse unitary transform, i.e., $\bm{x}$, and their continuous-time counterpart $x(t)$. 

\begin{defi}
The PAPR of DA-FrFDM is defined as

\begin{equation}\label{eq:PAPR_def}
    \eta\triangleq \frac{\max\big\{|x(t)|^2\big\}}{\mathbb{E}\big[|x(t)|^2\big]},~~~
    \eta_{\text{dB}}\triangleq 10\log\eta,
\end{equation}
where the continuous-time baseband DA-FrFDM signal can be written as
\begin{equation}\label{eq:x(t)}
\begin{aligned}
x(t) &= \sqrt{\frac{\sin\alpha + j\cos\alpha}{N}} \times \\
&\hspace{-1.4cm}\sum\limits_{k=0}^{N-1} s[k] e^{-\frac{j}{2}t^2\cot\alpha}  e^{-\frac{j}{2}k^2\cot\alpha(\Delta u)^2} e^{j2\pi kt/T},
\end{aligned}
\end{equation}
for $0\le t\le T$, where $T \triangleq N T_{s}$ is the block duration. By letting $t={nT}_s$ in \eqref{eq:x(t)}, one can obtain the discrete version of fractional Fourier transform in \eqref{eq:II1}.
\end{defi}

To derive $\eta$, we first analyze the envelope power function of $x(t)$, i.e., $|x(t)|^2$. Specifically, we have
\begin{eqnarray}\label{eq:II3}
&&\hspace{-0.4cm} |x(t)|^2 = x(t)\cdot x^*(t) = \frac{1}{N} \sum_{k=0}^{N-1} s[k]s^*[k] + \\
&&\hspace{-0.4cm} \frac{2}{N} \left\{ \sum_{p=1}^{N-1} \gamma_{p}^{(1)}(A_{\alpha})\cos(\frac{2\pi pt}{T}) + \sum_{p=1}^{N-1} \gamma_{p}^{(2)}(A_{\alpha})\cos(\frac{2\pi pt}{T}) \right. \notag\\
&&\hspace{-0.4cm} \quad \left. + \sum_{p=1}^{N-1} \gamma_{p}^{(3)}(A_{\alpha})\sin(\frac{2\pi pt}{T}) + \sum_{p=1}^{N-1} \gamma_{p}^{(4)}(A_{\alpha})\sin(\frac{2\pi pt}{T}) \right\} \notag\\
&&\hspace{-0.4cm} \triangleq \frac{1}{N} \sum_{k=0}^{N-1} s[k]s^*[k] + \frac{2}{N} \{g(t)\}, \notag
\end{eqnarray}
where we have defined
\begin{equation*}
\gamma_{p}^{(1)}(A_{\alpha}) \triangleq 
\sum\limits_{m=0}^{N-1-p}\lambda_{m,p}\cos\beta_{m,p}(A_{\alpha}),
\end{equation*}
\begin{equation*}
\gamma_{p}^{(2)}(A_{\alpha})
\triangleq -\sum\limits_{m=0}^{N-1-p}\mu_{m,p}\sin\beta_{m,p}(A_{\alpha}),
\end{equation*}
\begin{equation*}
\gamma_{p}^{(3)}(A_{\alpha})\triangleq -\sum\limits_{m=0}^{N-1-p}\lambda_{m,p}\sin\beta_{m,p}(A_{\alpha}),
\end{equation*}
\begin{equation*}
\gamma_{p}^{(4)}(A_{\alpha}) \triangleq -\sum\limits_{m=0}^{N-1-p}\mu_{m,p}\cos\beta_{m,p}(A_{\alpha}),
\end{equation*}
\begin{equation*}
\beta_{m,p}(A_{\alpha}) \triangleq p(2m+p)A_\alpha,~A_\alpha \triangleq -\frac{\pi^2 \sin(2\alpha)}{T^2},
\end{equation*}
\begin{equation*}
\lambda_{m,p} \triangleq \Re\{s[m+p]s^*[m]\},~ \mu_{m,p} \triangleq  \Im\{s[m+p]s^{*}[m]\}.
\end{equation*}



An important observation from \eqref{eq:II3} is that the angle $\alpha$ impacts PAPR through $A_\alpha$. This allows us to represent PAPR purely as a function of $A_\alpha$, i.e.,  $\eta(A_\alpha)$, rather than as a function of both $A_\alpha$ and $\alpha$. This simplification is crucial for effectively optimizing the PAPR.
From this foundation, we establish two properties on the PAPR of $x(t)$.

\begin{lem}\label{lem:1}
The PAPR of DA-FrFDM can be refined as
\begin{equation*}\label{eq:II6}
\begin{aligned}
\eta_{\text{dB}}&= 10\log\left(1+\frac{2\max\{g(t)\}}{\sum\limits_{k=0}^{N-1}s[k]s^*[k]}\right).
\end{aligned}
\end{equation*}
\end{lem}

\begin{proof}
The expectation of $g(t)$ in \eqref{eq:II3} is zero because $g(t)$ is a sum of trigonometric functions over a period $T$. Therefore,
\begin{equation}\label{eq:xtmeanpower}
\mathbb{E}\big[|x(t)|^2\big]=\frac1N\sum_{k=0}^{N-1}s[k]s^*[k]. 
\end{equation} 
Lemma \ref{lem:1} follows by substituting \eqref{eq:xtmeanpower} into \eqref{eq:PAPR_def}.
\end{proof}

\begin{thm}\label{thm:1}
The PAPR is a periodic function of $A_\alpha$ with a period of $\pi$. This periodicity suggests that within the range $\alpha\in[\frac{\pi}{2},\frac{\pi}{2}+\frac{1}{2}{\sin}^{-1}(\frac{T^2}{\pi}))$, all possible values of PAPR can be explored. This range effectively encapsulates the dynamics of PAPR variation with respect to the angle.
\end{thm}

\begin{proof}
Lemma~\ref{lem:1} indicates that analyzing the periodicity of $\max\{g(t)\}$ is sufficient for determining the periodic behavior of PAPR for a given data block $\bm{s}$. 

To demonstrate the periodicity, define $A_{\widetilde{\alpha}} \triangleq A_\alpha + \pi$ and substitute $A_{\widetilde{\alpha}}$ into $\gamma_{p}^{(i)}(A_\alpha)$, $i=1,2,3,4$. This substitution results in the following relation:
\begin{equation*}\label{eq:Thm2_1}
\gamma_{p}^{(i)}(A_{\widetilde{\alpha}})=(-1)^p\gamma_{p}^{(i)}(A_\alpha).
\end{equation*}

Additionally, it can be shown that
\begin{equation*}\label{eq:Thm2_2}
(-1)^{p}\cos\left(\frac{2\pi pt}{T}\right)=\cos\left(\frac{2\pi p(t+T/2)}{T}\right),
\end{equation*}
\begin{equation*}\label{eq:Thm2_3}
(-1)^{p}\sin\left(\frac{2\pi pt}{T}\right)=\sin\left(\frac{2\pi p(t+T/2)}{T}\right).
\end{equation*}

Let $\widetilde{g}(t)$ represent the value of $g(t)$ when $A_{\alpha}=A_{\widetilde{\alpha}}$. Using the above equations, we can express $\widetilde{g}(t)$ as follows:
\begin{eqnarray*}
    &&\hspace{-1cm} \widetilde{g}(t) = \sum_{i=1,2}\sum_{p=1}^{N-1}\gamma_{p}^{(i)}(A_{{\alpha}})\cos\biggl(\frac{2\pi p(t+T/2)}{T}\biggr) + \\
    && \sum_{j=3,4}\sum_{p=1}^{N-1}\gamma_{p}^{(j)}(A_{{\alpha}})\sin\biggl(\frac{2\pi p(t+T/2)}{T}\biggr).
\end{eqnarray*}

Since the sine and cosine terms in $\widetilde{g}(t)$ are periodic functions of $t$ with period $T$, shifting these functions along the $t$-axis by $T/2$ does not alter their maximum values. Thus,
\begin{equation*}
    \max\{\widetilde{g}(t)\} = \max\{g(t)\}.
\end{equation*}
This confirms that PAPR is a periodic function of $A_\alpha$ with a period of $\pi$, implying that all possible PAPR values can be examined within the interval $A_\alpha \in (b, b + \pi]$ for any $b \in \mathbb{R}$.


When $\alpha = \frac{\pi}{2}$, our DA-FrFDM degenerates to an OFDM system. Therefore, we can set $b = 0$, in which case $\alpha$ is bounded within the range $\alpha \in \left[\frac{\pi}{2}, \frac{\pi}{2} + \frac{1}{2} \sin^{-1}\left(\frac{T^2}{\pi}\right)\right)$.
\end{proof}

\subsection{Determination of the optimal angle}
Based on Lemma \ref{lem:1} and Theorem \ref{thm:1}, this section investigates how to determine the optimal angle, $\alpha^*$, that minimizes the PAPR. In practical communication systems, where the block duration $T$ is generally short, the PAPR exhibits significant sensitivity to small changes in $\alpha$. This sensitivity necessitates a finely tuned search over $\alpha$ to capture the nuances in PAPR behavior. To facilitate this process without incurring excessive computational costs, it is essential to minimize the search range for $\alpha$. 


The primary challenge in optimizing PAPR lies in managing the $\max{\left\{g(t)\right\}}$ operation. Typically, the $\max$ operation is considered an infinite norm, but for practical purposes, it can be closely approximated by the $n$-th root of the integral of $|g(t)|^n$ over time $T$, $\sqrt[n]{\int_{0}^{T}|g(t)|^ndt}$, even when $n$ is not particularly large \cite{closedform}. In this light, we opt for a surrogate function to simplify the analysis of $\max{\left\{g(t)\right\}}$. 
Specifically, we define the surrogate function as
\begin{equation}\label{eq:III16}
I=\int_0^Tg(t)^4dt.
\end{equation}

Our goal is to determine the optimal angle $\alpha$ for achieving the lowest PAPR. Therefore, we compute the derivative of $I$ with respect to (w.r.t.) $\alpha$:
\begin{equation}\label{eq:III17}
\begin{aligned}
I'(\alpha)&=\frac d{d\alpha}\int_0^Tg(t)^4dt = \frac d{dA_\alpha}\int_0^Tg(t)^4dt\cdot\frac{dA_\alpha}{d\alpha} \\
&= \frac d{dA_\alpha}\int_0^Tg(t)^4dt\cdot\Big(-\frac{2\pi^2\cos(2\alpha)}{T^2}\Big).
\end{aligned}
\end{equation}

The first term of \eqref{eq:III17} can be further refined as
\begin{equation}\label{eq:gt_Aa}
\begin{aligned}
&\frac{d}{dA_\alpha}\int_0^T g(t)^4\, dt \\
&= 4 \int_0^{2\pi} \left( \sum_{p=1}^{N-1} \gamma_{p }^{(1)}(A_{\alpha})\cos(pt) + \sum_{p=1}^{N-1} \gamma_{p }^{(2)}(A_{\alpha})\cos(pt) \right. \\
&+ \left. \sum_{p=1}^{N-1} \gamma_{p }^{(3)}(A_{\alpha})\sin(pt) + \sum_{p=1}^{N-1} \gamma_{p }^{(4)}(A_{\alpha})\sin(pt) \right)^3 \\
&\times \left( \sum_{p=1}^{N-1} \rho_{p }^{(1)}(A_{\alpha})\cos(pt) + \sum_{p=1}^{N-1} \rho_{p }^{(2)}(A_{\alpha})\cos(pt) \right. \\
&+ \left. \sum_{p=1}^{N-1} \rho_{p }^{(3)}(A_{\alpha})\sin(pt) + \sum_{p=1}^{N-1} \rho_{p }^{(4)}(A_{\alpha})\sin(pt) \right) dt.
\end{aligned}
\end{equation}
where 
\begin{equation*}
\rho_{p }^{(1)}(A_{\alpha})\triangleq -\sum\limits_{m=0}^{N-1-p}p(2m+p)\lambda_{m,p}\sin\beta_{m,p }(A_{\alpha}),
\end{equation*}
\begin{equation*}
\rho_{p }^{(2)}(A_{\alpha})\triangleq-\sum\limits_{m=0}^{N-1-p}p(2m+p)\mu_{m,p}\cos\beta_{m,p }(A_{\alpha}),
\end{equation*}
\begin{equation*}
\rho_{p }^{(3)}(A_{\alpha})\triangleq-\sum\limits_{m=0}^{N-1-p}p(2m+p)\lambda_{m,p}\cos\beta_{m,p }(A_{\alpha}),
\end{equation*}
\begin{equation*}
\rho_{p }^{(4)}(A_{\alpha})\triangleq\sum\limits_{m=0}^{N-1-p}p(2m+p)\mu_{m,p}\sin\beta_{m,p }(A_{\alpha}).
\end{equation*}

Upon first examination, \eqref{eq:gt_Aa} appears computationally intensive due to its inclusion of $256\left(N-1\right)^4$ individual integrals, all of which are initially perceived as complex due to their structure. However, these integrals can be efficiently computed by leveraging the periodic nature of trigonometric functions.

\begin{thm}\label{thm:tri}
Let $\xi_1(t), \xi_2(t), \xi_3(t), \xi_4(t) \in \{\cos(t), \sin(t)\}$.
For indices $1\le k,l,m,n\le N-1$,  the integral of the product of these trigonometric functions over a full period is given by
\begin{equation*}
\int_0^{2\pi} \xi_1(kt)\xi_2(lt)\xi_3(mt)\xi_4(nt)\, dt =  \frac{\pi}{4} \bm{q}_i
\begin{bmatrix}
\begin{smallmatrix}
\delta(k + l + m + n)\\
\delta(k + l - m - n)\\
\delta(k + l + m - n)\\
\delta(k + l - m + n)\\
\delta(k - l + m + n)\\
\delta(k - l - m - n)\\
\delta(k - l + m - n)\\
\delta(k - l - m + n)
\end{smallmatrix}
\end{bmatrix},
\end{equation*}
where $\bm{q}_i$ is the $i^{th}$ row of matrix $\bm{Q}$:
\begin{eqnarray*}
i = &&\hspace{-0.5cm} 8\delta\big(\xi_1(t)\!=\!\sin(t)\big) + 4\delta\big(\xi_2(t)\!=\!\sin(t)\big) + \\
&&\hspace{-0.5cm}  2\delta\big(\xi_3(t)\!=\!\sin(t)\big) + \delta\big(\xi_4(t)\!=\!\sin(t)\big) + 1,
\end{eqnarray*}
\begin{equation*}\label{Q}
\bm{Q} = \begin{pmatrix}
\begin{smallmatrix}
0 & +1 & +1 & +1 & +1 & +1 & +1 & +1 \\
0 & 0 & 0 & 0 & 0 & 0 & 0 & 0 \\
0 & 0 & 0 & 0 & 0 & 0 & 0 & 0 \\
0 & -1 & +1 & +1 & -1 & -1 & +1 & +1 \\
0 & 0 & 0 & 0 & 0 & 0 & 0 & 0 \\
0 & +1 & +1 & -1 & +1 & -1 & -1 & +1 \\
0 & +1 & -1 & +1 & +1 & -1 & +1 & -1 \\
0 & 0 & 0 & 0 & 0 & 0 & 0 & 0 \\
0 & 0 & 0 & 0 & 0 & 0 & 0 & 0 \\
0 & +1 & +1 & -1 & -1 & +1 & +1 & -1 \\
0 & +1 & -1 & +1 & -1 & +1 & -1 & +1 \\
0 & 0 & 0 & 0 & 0 & 0 & 0 & 0 \\
0 & -1 & -1 & -1 & +1 & +1 & +1 & +1 \\
0 & 0 & 0 & 0 & 0 & 0 & 0 & 0 \\
0 & 0 & 0 & 0 & 0 & 0 & 0 & 0 \\
0 & +1 & -1 & -1 & -1 & -1 & +1 & +1
\end{smallmatrix}
\end{pmatrix}.
\end{equation*}
\end{thm}


\begin{proof} 
Using the product-to-sum formula, we can express the product of $\xi_1\left(kt\right)$, $\xi_2\left(lt\right)$, $\xi_3\left(mt\right)$, and $\xi_4\left(nt\right)$ as a sum (or difference) of eight trigonometric functions. 
As an example,
\begin{equation*}
\begin{aligned}
&\cos(2t)\cdot\cos(4t)\cdot\cos(1t)\cdot\cos(5t)\\
&=\left(\frac{1}{2}[\cos(6t)+\cos(-2t)]\right)\cdot\left(\frac{1}{2}[\cos(6t)+\cos(-4t)]\right)\\
&=\frac{1}{4}[\cos(6t)\cos(6t)+\cos(6t)\cos(-4t)\\
&\hspace{0.8cm} +\cos(-2t)\cos(6t)+\cos(-2t)\cos(-4t)]\\
&=\frac{1}{8}[\cos(12t)+1+\cos(2t)+\cos(10t)\\
&\hspace{0.8cm}+\cos(4t)+\cos(-8t)+\cos(-6t)+\cos(2t)].
\end{aligned}
\end{equation*}

In general, we represent this process as:
\begin{equation*}
\begin{aligned}
& \xi_1(kt)\cdot\xi_2(lt)\cdot\xi_3(mt)\cdot\xi_4(nt)\\
& =\frac14[u_1((k+l)t)+u_{2}((k-l)t)]\\
&  \hspace{0.5cm} \times [u_{3}((m+n)t)+u_{4}((m-n)t)]\\
& =\frac18[u_5((k+l+m+n)t)+u_6((k+l-m-n)t)\\
&  \hspace{0.5cm} +u_7((k+l+m-n)t)+u_8((k+l-m+n)t)\\
&  \hspace{0.5cm} +u_9((k-l+m+n)t)+u_{10}((k-l-m-n)t)\\
&  \hspace{0.5cm} +u_{11}((k-l+m-n)t)+u_{12}((k-l-m+n)t)],
\end{aligned}
\end{equation*}
where $u_{1\sim 12}(t)$ represents either $\pm \cos(t)$ or $\pm \sin(t)$. 

In particular,
\begin{equation*}
   \int_0^{2\pi} u_{i}(wt) dt = 
   \begin{cases}
       0, & \text{If $u_{i}(wt)=\pm \sin(wt)$}, \\
       \pm 2\pi \delta(w), & \text{If $u_{i}(wt)=\pm \cos(wt)$}.
   \end{cases}
\end{equation*}

This outcome guides the computation using the matrix $\bm{Q}$. Specifically,
\begin{itemize}
    \item If there are either one or three $\cos(t) $ terms among $\xi_{1\sim4}(t)$, then  $u_{5\sim12}(t)$  will be   $\pm \sin(t) $, resulting in a zero integral. This corresponds to the zero rows in $\bm{Q}$.
    \item Otherwise, $ u_{5\sim12}(t)$ are $\pm \cos(t)$, and the integral's outcome depends on the frequency combinations, represented by the $\pm{1}$ entries in $\bm{Q}$.
    \item The first column of $\bm{Q}$ is entirely zero because $k + l + m + n$ is non-zero for $1 \leq k, l, m, n \leq N-1$.
\end{itemize}

The row index $i$ is determined by the indicator functions  $\delta(\xi_j(t) = \sin(t))$, which check if each  $\xi_j(t)$ is  $\sin(t)$ or $\cos(t)$.
\end{proof}

\begin{algorithm}[t]
\caption{Discover the optimal angle $\alpha^*$ for DA-FrFDM}\label{algo:1}
\begin{algorithmic}[1]
\State {\bf Input:} $N$, $T_s$, and $\bm{s}$.
\State {\bf Output:} Optimal angle $\alpha^{*}$.

\State Determine the initial search set $\Omega_0$ based on Theorem \ref{thm:1}.

\For{$\alpha_i \in \Omega_0$}
    \State Compute $\gamma_{p}^{(j)}(A_{\alpha_i})$ and $\rho_{p}^{(j)}(A_{\alpha_i})$ for $j = 1, 2, 3, 4$.
    \State Compute $I'(\alpha_i)$ by Theorem~\ref{thm:tri}.
\EndFor

\State Initialize $\Omega = \emptyset$.
\For{$\alpha_i\in\Omega_0$}
    \If{$I'(\alpha_i) \leq 0$ and $I'(\alpha_i + \Delta \alpha) \geq 0$} 
        \State $\Omega = \Omega \cup \left\{\alpha_i + j \Delta \alpha^{\prime} : j = 0, 1, 2, \dots, \frac{\Delta \alpha}{\Delta \alpha^{\prime}}\right\}$.
    \EndIf
\EndFor
 
\For{$\alpha_j\in\Omega$}
    \If{$I'(\alpha_j) \leq 0$ and $I'(\alpha_j + \Delta \alpha^{\prime}) \geq 0$}
        \State Retain $\alpha_j$ in $\Omega$.
    \Else
        \State $\Omega=\Omega\backslash\{\alpha_j\}$.
    \EndIf
\EndFor

\State Search through $\Omega$: $\alpha^*=\arg\min_{\alpha\in\Omega}\eta(A_{\alpha})$.
\end{algorithmic}
\end{algorithm}

Theorem~\ref{thm:tri} provides an efficient method for computing $I'(\alpha)$, which underpins the design of an algorithm to identify a refined set of candidate angles, $\Omega$, for locating the optimal angle $\alpha^*$ that minimizes PAPR. The algorithm is described in Algorithm~\ref{algo:1}, with details on its steps outlined below.

The algorithm begins by taking three input parameters: the number of subcarriers $N$, the sampling period $T_s$, and the signal block $\bm{s}$. Using these inputs, it establishes an initial search range for $\alpha$ based on Theorem \ref{thm:1}: 
$\alpha \in \left[\frac{\pi}{2}, \frac{\pi}{2} + \frac{1}{2}\sin^{-1}\left(\frac{T^2}{\pi}\right)\right)$.
This range is then discretized with an initial step size $\Delta \alpha$, generating an initial search set:
\begin{equation}\label{eq:init}
\Omega_0 \!=\! \left\{\frac{\pi}{2} + i\Delta \alpha : i = 0, 1, 2, \ldots, \frac{1}{2\Delta \alpha}\sin^{-1}\!\left(\frac{T^2}{\pi}\right) \!-\! 1 \right\}.
\end{equation}

For each $\alpha_i \in \Omega_0$, the algorithm computes $\gamma_{p}^{(j)}(A_{\alpha_i})$ and $\rho_{p}^{(j)}(A_{\alpha_i})$ for $j = 1, 2, 3, 4$, along with the corresponding $I'(\alpha_i)$ using Theorem~\ref{thm:tri}. It then identifies potential local minima by verifying if $I'(\alpha_i) \leq 0$ and $I'(\alpha_i + \Delta \alpha) \geq 0$. 
When these conditions are met, indicating a local minimum, we form a finer set of candidate values:
\begin{equation*}
\Omega = \Omega \cup \bigg\{\alpha_i + j\Delta \alpha^{\prime}, j = 0, 1, 2, \ldots, \frac{\Delta \alpha}{\Delta \alpha^{\prime}}\bigg\},
\end{equation*}
where $\Delta \alpha^{\prime}$ is a finer step size, and $\Omega$ is initialized as an empty set.
Next, the algorithm further narrows the finer search set $\Omega$ by iterating over each candidate $\alpha_j\in \Omega$ and retaining only those values that satisfy $I'(\alpha_j) \leq 0$ and $I'(\alpha_j + \Delta \alpha^{\prime}) \geq 0$.

In the final step, the algorithm searches over $\Omega$ to identify the optimal angle that minimizes the PAPR.

\section{Circular Convolution Theorem for DA-FrFDM}\label{sec:IV}

In multi-carrier systems, the circular convolution theorem serves as a foundation for simplifying equalization \cite{stuber2004broadband,PAMA}. In traditional OFDM systems, the circular convolution theorem allows for straightforward equalization by converting the channel's time-domain impulse response to the frequency domain, where it acts as a simple one-tap equalizer. 

When extending to FrFT-based multi-carrier systems, one-tap equalization requires transforming the channel into the fractional Fourier domain \cite{amir}. However, this transformation complicates the receiver since it demands computing the fractional Fourier domain response of the channel. This added complexity makes direct application of the fractional Fourier convolution theorem cumbersome for systems that vary 
$\alpha$, as each angle change requires recalculating the channel response.

To address the limitations of conventional FrFT-based equalization in the DA-FrFDM system, we propose a refined circular convolution theorem. This theorem leverages a quadratic phase term applied to the time-domain samples, enabling one-tap equalization without recalculating the channel for each fractional angle. By introducing this phase term, DA-FrFDM can achieve effective equalization using the fixed frequency domain channel response, independent of the fractional angle.

\begin{figure*}[t]
  \centering
  \includegraphics[width=1.9\columnwidth]{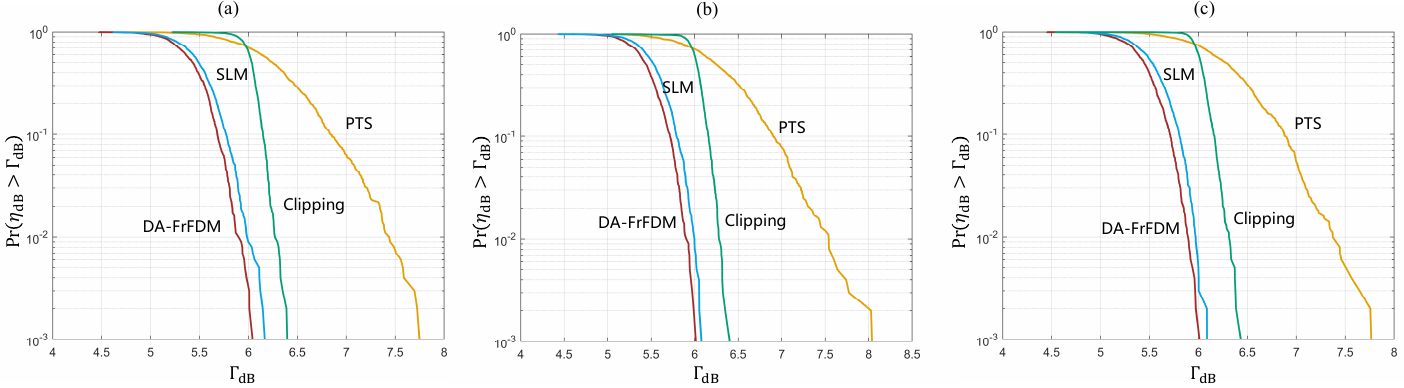}\\
  \caption{CCDF of PAPR with different PAPR reduction techniques: (a) PAPR reduction for complex Gaussian signals; (b) PAPR reduction for 64QAM; (c) PAPR reduction for 128QAM.}
\label{fig:papr}
\end{figure*}

\begin{thm}[Circular Convolution for DA-FrFDM] \label{thm:cir_thm}
In DA-FrFDM, let $\bm{z}$ denote the time-domain samples modified by the phase term $e^{j\theta(n)}$ at the transmitter, and $\widehat{\bm{z}}$ the time-domain samples at the receiver after CP removal. The samples $\widehat{\bm{z}}$ can be expressed as
\begin{equation*}
\widehat{z}[n]=h[n]\circledast z[n], n=0, 1,\cdots, N-1,
\end{equation*}
where $\circledast$ represents circular convolution, and $\bm{h}=\big[h[0],\allowbreak h[1],\allowbreak \cdots,\allowbreak h[N-1]\big]^\top$ denotes the time-domain samples of channel impulse response.

Let $\theta(n)=\frac{1}{2}n^2\cot\alpha\allowbreak T_s^2$, we have
\begin{equation*}
\widehat{s}[k] = h_f[k]s[k], k=0, 1,\cdots, N-1,
\end{equation*}
where $\bm{h}_f=\bm{F}_{1}\bm{h}$ is the frequency-domain channel response  independent of the fractional angle $\alpha$.
\end{thm}

\begin{proof}
From the definition, we have $z[n] = x[n] \allowbreak e^{\frac{j}{2}n^2\cot\alpha T_s^2}$ and $\widehat{x}[n] = \widehat{z}[n] e^{-\frac{j}{2}n^2\cot\alpha T_s^2}$.

Let $k_\alpha=\sqrt{\frac{\sin\alpha + j\cos\alpha}{N}}$, we have
\begin{eqnarray*}
&&\hspace{-0.4cm} \widehat{s}[k] = k_\alpha^* \sum_{n=0}^{N-1} h[n]\circledast\widetilde{x}[n] e^{-\frac{j}{2}n^2\cot\alpha T_s^2}  \times\\
&&\hspace{0.6cm}  e^{\frac{j}{2}n^2\cot\alpha(T_s L)^2} e^{\frac{j}{2}k^2\cot\alpha(\Delta u)^2} e^{-j2\pi nk/N} \\
&&\hspace{-0.4cm} = k_\alpha^* \sum_{n=0}^{N-1} h[n]\circledast\widetilde{x}[n] e^{\frac{j}{2}k^2\cot\alpha(\Delta u)^2} e^{-j2\pi nk/N}.
\end{eqnarray*}
After expanding the circular convolution,
\begin{eqnarray*}
&&\hspace{-0.4cm} \widehat{s}[k]= k_\alpha^* \sum_{m=0}^{N-1} \left( \sum_{n=m}^{N-1} h[m] x[n-m] \times \right.\\
&&\hspace{-0.4cm}  \left.  e^{\frac{j}{2}(n-m)^2\cot\alpha T_s^2} e^{\frac{j}{2}k^2\cot\alpha(\Delta u)^2} e^{-j2\pi nk/N} \right)+\\
&&\hspace{-0.4cm}   k_\alpha^* \sum_{m=0}^{N-1} \left( \sum_{n=0}^{m-1} h[m] x[n-m+N] \right. \times\\
&&\hspace{-0.4cm}  \left.e^{\frac{j}{2}(n-m+N)^2\cot\alpha T_s^2} e^{\frac{j}{2}k^2\cot\alpha(\Delta u)^2} e^{-j2\pi nk/N} \right) 
\end{eqnarray*}
Following a change of variables, we obtain
\begin{eqnarray*}
&&\hspace{-1.4cm} \widehat{s}[k]= k_\alpha^* \sum_{m=0}^{N-1} \left( \sum_{p=0}^{N-m-1} h[m] x[p] e^{\frac{j}{2}p^2\cot\alpha T_s^2} \right.\times \\
&&\hspace{-0.4cm}  \left.  e^{\frac{j}{2}k^2\cot\alpha(\Delta u)^2} e^{-j2\pi (m+p)k/N} \right) +\\
&&\hspace{-0.4cm}   k_\alpha^* \sum_{m=0}^{N-1} \left( \sum_{p=N-m}^{N-1} h[m] x[p] e^{\frac{j}{2}p^2\cot\alpha T_s^2} \right.\times\\
&&\hspace{-0.4cm}  \left.  e^{\frac{j}{2}k^2\cot\alpha(\Delta u)^2} e^{-j2\pi (m+p-N)k/N} \right)
\end{eqnarray*}
Finally, the equation can be further refined as
\begin{eqnarray*}
&&\hspace{-1.5cm} \widehat{s}[k]= k_\alpha^* \sum_{m=0}^{N-1} \sum_{p=0}^{N-1} h[m] x[p] e^{\frac{j}{2}p^2\cot\alpha T_s^2}  \times\\
&&\hspace{-0.9cm}  e^{\frac{j}{2}k^2\cot\alpha(\Delta u)^2} e^{-j2\pi (m+p)k/N} \\
&&\hspace{-0.8cm}  =  h_f[k] s[k],
\end{eqnarray*}
proving Theorem \ref{thm:cir_thm}.
\end{proof}

The circular convolution theorem thus allows DA-FrFDM to leverage fractional Fourier-based modulation to optimize PAPR with minimal computational overhead for the one-tap equalization at the receiver.

\begin{rem}
Theorem \ref{thm:cir_thm} establishes a foundational form of the circular convolution theorem for DA-FrFDM. In practical implementations, as well as in the simulations in Section \ref{sec:V}, time-domain signals are often oversampled to accurately approximate the analog waveform \cite{shao2021federated}. Given an oversampling rate of $L$, oversampling can be achieved by zero-padding $\bm{s}$ with $NL-L$ additional points and performing an $NL$-point IDFrFT.

For oversampled DA-FrFDM systems, the circular convolution theorem remains valid if we define the phase term as $\theta(i)=\frac{1}{2}i^2\cot\alpha T^2_s L^2$, where $i=0,1,\cdots,NL-1$. The proof follows the steps in the proof of Theorem \ref{thm:cir_thm} and is therefore omitted for brevity.
\end{rem}

\section{Simulation Results}\label{sec:V}
\begin{table}[t]
\caption{Simulation parameter settings.}
\label{tab:notations}
\centering
\setlength{\tabcolsep}{3mm} 
\begin{tabular}{ccc}
\toprule
\textbf{Parameters} & \textbf{Descriptions} &\textbf{Value} \\ 
\midrule
$N$ &Number of subcarriers& $64$ \\ 
$N_{cp}$ &CP length& $10$ \\
$L$ &Oversampling rate& $10$ \\ 
$T$ &Symbol duration& $128~\mu s$ \\ 
$\Delta\alpha$ &Step size& $\Delta{\alpha_1} = \frac{1}{80} \sin^{-1}\left(\frac{T^2}{\pi}\right)$ \\  
$\Delta\alpha^\prime$ &Step size& $\frac{\Delta{\alpha_1}}{39}$ \\ 
$\text{CR}$ &Clipping ratio& $\text{CR} = 2$ \\ 
\bottomrule
\end{tabular}
\label{tab:params}
\end{table}

\begin{figure*}[t]
  \centering
  \includegraphics[width=1.9\columnwidth]{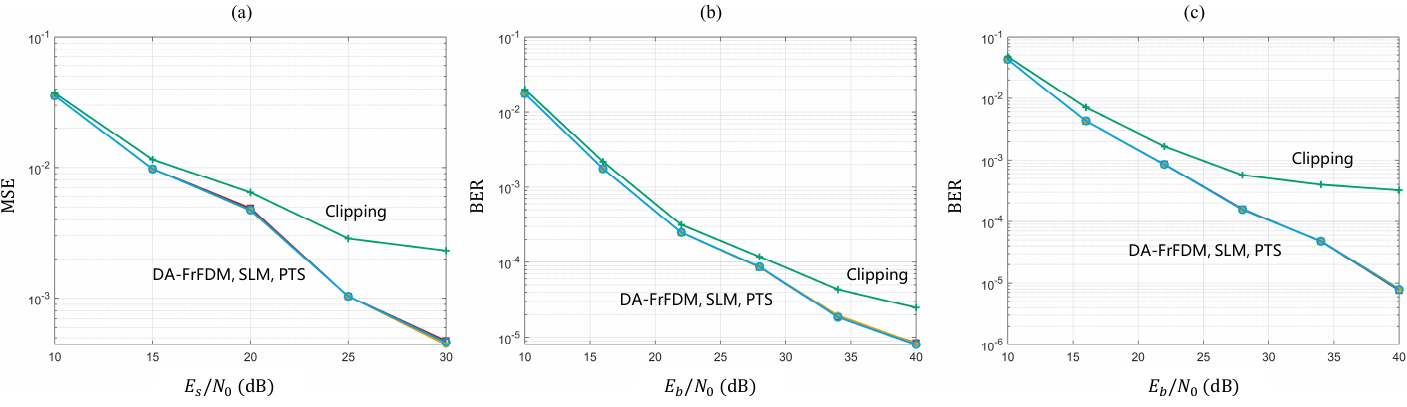}\\
  \caption{Decoding performance with different PAPR reduction techniques: (a) MSE of complex Gaussian signals; (b) BER of 64QAM signals; (c) BER of 128QAM signals.}
\label{fig:mse_ber}
\end{figure*}

In this section, we evaluate the performance of our DA-FrFDM system, focusing on three key advantages: reduced PAPR, simplified equalization, and resilience to ICI in doubly dispersive channels.

\subsection{PAPR Performance}
To assess the PAPR performance of our DA-FrFDM scheme, we selected three widely-used PAPR reduction methods for OFDM systems as benchmarks: Clipping, Selective Mapping (SLM), and Partial Transmit Sequence (PTS). These methods were chosen for comparison because, like DA-FrFDM, they preserve spectral efficiency and require minimal signaling overhead.

The simulation parameters are summarized in Table~\ref{tab:params}. Specifically, the number of subcarriers $N = 64$, the CP length $N_{cp}=10$, and the symbol duration $T = 128~\mu s$. To accurately capture PAPR behavior of analog waveforms, we consider an oversampling rate $L = 10$. The data symbols 
$\bm{s}$ are transmitted as either complex Gaussian symbols or QAM modulated symbols to evaluate the versatility of DA-FrFDM across different symbol types.

The Complementary Cumulative Distribution Function (CC\-DF) is used as the metric for PAPR performance, reflecting the probability that the PAPR exceeds a given threshold: $\text{CCDF}(\Gamma) =  \Pr(\eta > \Gamma)$. This measure provides a comprehensive view of peak occurrence in the transmitted signal, allowing a direct comparison across schemes.

Fig.~\ref{fig:papr} illustrates the CCDF performance for various PAPR reduction schemes, with data symbols configured as either complex Gaussian or QAM. In the baseline OFDM system without PAPR reduction, the PAPR required to reach a CCDF of $10^{-3}$ is approximately $12.3$ dB. In contrast, our DA-FrFDM scheme significantly reduces PAPR, nearly halving it compared to the baseline. When benchmarked against leading PAPR reduction techniques, i.e., clipping, SLM, and PTS, DA-FrFDM consistently outperforms, confirming the effectiveness of fractional Fourier domain-based PAPR reduction.

\begin{rem}
For DA-FrFDM, Algorithm~\ref{algo:1} provides an efficient way to identify the search set $\Omega$ for the optimal angle. However, each candidate within $\Omega$ still requires PAPR evaluation to determine the angle that minimizes PAPR. By comparison, SLM and PTS techniques achieve PAPR reduction by applying a phase sequence to the signal, which similarly involves searching through phase sequences and evaluating PAPR to find the optimal configuration. To ensure a fair comparison, in Fig.~\ref{fig:papr} we set a fixed PAPR evaluation budget of $128$ for all methods.
\end{rem}

\subsection{Equalization efficiency}
To evaluate the equalization performance of the DA-FrFDM system, we assess its decoding accuracy, providing validation for the improved circular convolution theorem as stated in Theorem \ref{thm:cir_thm}.

The simulation results are illustrated in Fig.~\ref{fig:mse_ber}, where a six-path Rayleigh fading channel is modeled. For the QAM data symbols, bit error rate (BER) performance is analyzed, while mean square error (MSE) is assessed for Gaussian symbols.

As can be seen, the DA-FrFDM system achieves decoding performance comparable to OFDM systems using SLM and PTS techniques under Rayleigh fading conditions, with the added benefit of lower PAPR. This outcome confirms the efficiency and simplicity of the proposed equalization process. In contrast, the clipping method leads to a noticeable decline in decoding performance, underscoring the advantages of the DA-FrFDM approach.

\begin{figure}[t]
  \centering
  \includegraphics[width=0.8\columnwidth]{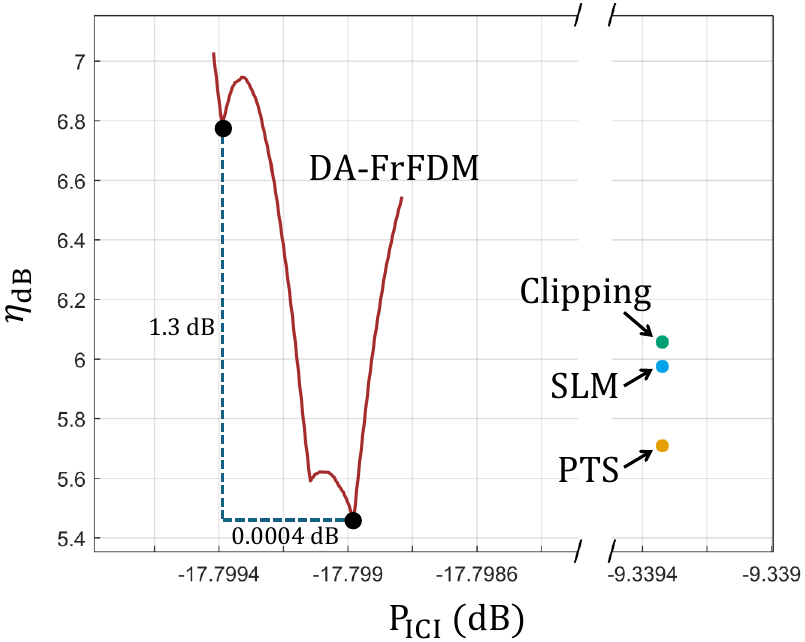}\\
  \caption{Trade-off between PAPR and $P_{ICI}$ in doubly dispersive channels.}
\label{fig:ici}
\end{figure}

\subsection{ICI mitigation in doubly dispersive channels}
In addition to reducing the PAPR, the DA-FrFDM system offers an inherent advantage of mitigating ICI in doubly dispersive channels, a major advantage over traditional OFDM systems. In these channels -- characterized by frequency-selective and fast fading -- DA-FrFDM adjusts the frequency of each subcarrier over time at a rate proportional to $\cot{\alpha}$, allowing it to adapt to channel variations more effectively than static OFDM.

The ICI mitigation capability of DA-FrFDM is closely tied to the angle $\alpha$ in the fractional Fourier domain \cite{mokhtari2015near}. This parameter influences the system's adaptability in doubly dispersive channels, enabling DA-FrFDM to perform dual roles: reducing PAPR while also enhancing ICI resilience. The ability to adjust 
$\alpha$ allows DA-FrFDM to dynamically balance PAPR reduction with ICI suppression based on channel conditions.

An intriguing question arises regarding the potential trade-off between PAPR reduction and ICI mitigation in DA-FrFDM. Specifically, adjusting $\alpha$ to optimize PAPR performance may influence the system's ICI handling capabilities, inviting further exploration into the optimal balance between these two performance metrics. This flexibility positions DA-FrFDM as a robust, adaptable solution for challenging wireless environments.

Consider a doubly dispersive channel with four distinct paths characterized by path gains of $0$ dB, $-4$ dB, $-5$ dB, and $-8$ dB, Doppler shifts of $500Hz$, $1600Hz$, $2200Hz$, and $3800Hz$, and path delays of $0$, $10 \mu s$, $20 \mu s$, and $40 \mu s$, respectively. Fig.~\ref{fig:ici} illustrates the trade-off between PAPR and the power of ICI for the DA-FrFDM system, specifically for a complex Gaussian data symbol block realization.

An important observation from the figure is that by adjusting the angle parameter, 
$\alpha$, the PAPR of the transmitted signal can be incrementally decreased. Notably, this adjustment results in only a minimal increase in ICI. This outcome demonstrates a promising feature of DA-FrFDM, indicating that the system can reduce PAPR effectively while preserving its capability for ICI mitigation, even in doubly dispersive channels.

\section{Conclusions}\label{sec:Conclusion}
This paper presented an innovative approach to tackling the high PAPR challenge in multi-carrier systems. Operating within the fractional Fourier domain and using a dynamic angle adjustment, DA-FrFDM achieves significant PAPR reduction and enhances ICI mitigation, even in doubly dispersive channel conditions. Additionally, DA-FrFDM enables efficient one-tap equalization by applying a quadratic phase sequence, simplifying signal processing and reducing receiver complexity.

Beyond PAPR reduction, the DA-FrFDM framework offers a versatile solution with potential for improved energy efficiency, enhanced reliability, and adaptability to challenging channel conditions in next-generation wireless communications. The insights from this work suggest that DA-FrFDM could serve as a robust foundation for future developments in multi-carrier communication systems, addressing both current and emerging performance demands in wireless networks.

\bibliographystyle{IEEEtran}
\bibliography{References}

\end{document}